\documentclass{article}
\usepackage{authblk}
\usepackage[utf8]{inputenc}
\usepackage{tabularx}
\usepackage{graphicx}
\usepackage{booktabs}
\usepackage{multirow}
\usepackage{todonotes}
\usepackage{amsmath}
\usepackage{url} 
\usepackage{tikz} 
\usepackage{tikzpeople} 
\usepackage{mwe}
\usepackage{amsthm}
\usepackage{amsmath}
\usepackage{amssymb}
\usepackage[ruled]{algorithm2e}
\usepackage{pifont}

\SetAlFnt{\small}
\SetAlCapFnt{\small}
\SetAlCapNameFnt{\small}
\SetAlCapHSkip{0pt}
\IncMargin{-\parindent}

\newtheorem{lemma}{Lemma}
\newcommand{\ohyes}{\ding{51}}
\newcommand{\ohno}{\ding{55}}
\newcommand{\ohpartial}{\ding{57}}
\AtBeginDocument{%
	\providecommand\BibTeX{{%
			\normalfont B\kern-0.5em{\scshape i\kern-0.25em b}\kern-0.8em\TeX}}}

\begin{document}

\title{Verifiable Fairness: Privacy--preserving Computation of Fairness for Machine Learning Systems}
\author[1]{Ehsan Toreini}
\author[2]{Maryam Mehrnezhad}
\author[3]{Aad Van Moorsel}
\affil[1]{University of Surrey, UK}
\affil[2]{Royal Holloway University of London}
\affil[3]{Birmingham University}
\date{}                     
\setcounter{Maxaffil}{0}
\renewcommand\Affilfont{\itshape\small}
\maketitle              
\begin{abstract}
 Fair machine learning is a thriving and vibrant research topic. In this paper, we propose Fairness as a Service (FaaS), a secure, verifiable and privacy-preserving protocol to computes and verify the fairness of any machine learning (ML) model. In the deisgn of FaaS, the data and outcomes are represented through cryptograms to ensure privacy. Also, zero knowledge proofs guarantee the well-formedness of the cryptograms and underlying data. FaaS is model--agnostic and can support various fairness metrics; hence, it can be used as a service to audit the fairness of any ML model. Our solution requires no trusted third party or private channels for the computation of the fairness metric. The security guarantees and commitments are implemented in a way that every step is securely transparent and verifiable from the start to the end of the process. The cryptograms of all input data are publicly available for everyone, e.g., auditors, social activists and experts, to verify the correctness of the process. We implemented FaaS to investigate performance and demonstrate the successful use of FaaS for a publicly available data set with thousands of entries.  

\end{abstract}

\section{Introduction}
\label{s:intro}
Demonstrating the fairness of algorithms is critical to the continued proliferation and acceptance of algorithmic decision making in general, and AI-based systems in particular. There is no shortage of examples that have diminished trust in algorithms because of unfair discrimination of groups within our population. This includes news stories about the human resource decision-making tools used by large companies, which turn out to discriminate against women \cite{that2018favored}. There also are well-understood seminal examples studied widely within the academic community, such as the unfair decisions related to recidivism in different ethnicities \cite{larson2016we}. In the UK, most recently the algorithm to determine A-levels substitute scores under COVID-19 was widely found to be unfair across demographics \cite{ALevel}. 

There has been a surge of research that aims to establish metrics that quantify the fairness of an algorithm.  This is an important area of research, and tens of different metrics have been proposed, from individual fairness to group fairness.  It has been shown that various expressions for fairness cannot be satisfied or optimised at once, thus establishing impossibility results~\cite{friedler2016possibility}. Moreover, even if one agrees about a metric, this metric on its own does not provide trust to people. It matters not only what the metrics express, but also who computes the metrics and whether one can verify these computations and possibly appeal against them.  At the same time, in situations in which verification by stakeholders is possible, the owner of the data wants to be assured that none of the original, typically sensitive and personal, data is leaked. The system that runs the algorithms (later referred to as Machine Learning system or ML system) may have a valid interest in maintaining the secrecy of the model.   
In other words, if one wants to establish {\em verifiable fairness}, one needs to tackle a number of security, privacy and trust concerns.  

In FaaS, we take a fundamentally different design approach. We leak no data or model information, but the FaaS is still able to calculate fairness for a variety of fairness metrics and independent of the ML model. Thus, replacing the model in the ML system will not impact functionality of FaaS protocol. Moreover, any other party can verify this calculation since all the necessary encrypted information is posted publicly, on a `fairness board'.  

Summarising, our contributions are: 
\begin{itemize}
\item We propose FaaS, a model--agnostic protocol to compute different fairness metrics without accessing sensitive information about the model and the dataset.


\item FaaS is universally verifiable so everyone can verify the well--formedness of the cryptograms and the steps of the protocol.

\item We implement a proof-of-concept of the FaaS architecture and protocol using off-the-shelf hardware, software, and datasets and run experiments to demonstrate the practical feasibility of FaaS.
\end{itemize}

\section{Background and Related Work}
\label{s:related}
One of the benefits of auditing ML-based products relates to {trust}. Trust and trustworthiness~(in socio-technical terms) are complicated matters. 
Toreini et. al~\cite{toreini2020relationship} proposed a framework for trustworthiness technologies in AI--solutions based on existing social frameworks on trust~(i.e. demonstration of Ability, Benevolence and Integrity, a.k.a. ABI and ABI+ frameworks) and technological trustworthiness~\cite{siau2018building}. They comprehensively reviewed the policy documents on regulating AI and the existing technical literature and derived any ML--based solution needs to demonstrate fairness, explainability, auditability, and safety and security to establish social trust. 
When using AI solutions, one cannot be assured of the fairness of such systems without trusting the reputation of the technology provider (e.g., datasets and ML models). It is commonly believed that leading tech companies do not make mistake in their implementation~\cite{e2020pathways}; however, in practice, we often witness that such products indeed suffer from bias in ML~\cite{that2018favored,ALevel}.

\begin{table}[t]
\centering
\caption{Features of FaaS and comparison with other privacy--oriented fair ML proposals~(support: full: \ohyes, partial: \ohpartial, none: \ohno)}
\resizebox{0.8\textwidth}{!}{%
\begin{tabular}{@{}l|c|c|c|c|c|c@{}}
\toprule
 Work       & \begin{tabular}[c]{@{}c@{}}Universal\\ Verifiability\end{tabular} & \begin{tabular}[c]{@{}c@{}}Ind. of \\ metric\end{tabular} & \begin{tabular}[c]{@{}c@{}}Ind. of\\ ML model\end{tabular} & \begin{tabular}[c]{@{}c@{}}User \\ Privacy\end{tabular} & \begin{tabular}[c]{@{}c@{}}Model\\ Confidentiality\end{tabular} & \begin{tabular}[c]{@{}c@{}}Off-the--shelf\\ Hardware\end{tabular}\\ \midrule
Veal \& Binns~\cite{veale2017fairer} &\ohno&\ohno&\ohno&\ohno&\ohno&\ohyes \\
Kilbertus et al.~\cite{kilbertus2018blind}   &\ohpartial& \ohno & \ohno&\ohyes&\ohyes&\ohyes \\
Jagielski et al.~\cite{jagielski2019differentially} & \ohno & \ohno & \ohno & \ohyes & \ohno&\ohyes \\
Hu et al.~\cite{hu2019distributed}   & \ohno & \ohno & \ohno & \ohyes & \ohno&\ohyes \\
Segal et al.~\cite{segal2021fairness}   & \ohpartial & \ohyes & \ohyes & \ohyes & \ohyes&\ohyes \\
Park et al.~\cite{park2022fairness}   & \ohpartial & \ohyes & \ohyes & \ohyes & \ohyes &\ohno\\
\midrule
FaaS (this paper)   &\ohyes&\ohyes&\ohyes&\ohyes&\ohyes&\ohyes\\
\bottomrule
\end{tabular}
\label{tbl:comparison}
}
\end{table}

\subsection{Fairness Metrics}
\label{ss:fairness}
There exist several fairness definitions in the literature. Designing a fair algorithm requires measuring and assessment of fairness. Researchers have worked on formalising fairness for a long time. Narayanan~\cite{narayanan2018translation} lists at least 21 different fairness definitions in the literature and this number is growing, e.g., \cite{chouldechova2017fair,corbett2017algorithmic}. 

Fairness is typically expressed as discrimination in relation to data features.  These features for which discrimination may happen are known as \emph{Protected Attributes} (PAs) or sensitive attributes. These include, but are not limited to, ethnicity, gender, age, scholarity, nationality, religion and socio-economic group. 

The majority of fairness definitions expresses fairness in relation to PAs.  
In this paper, we consider Group Fairness, which refers to a family of definitions, all of which consider the performance of a model on the population groups level. The fairness definitions in this group are focused on keeping decisions consistent across groups and are relevant to both disparate treatment and disparate impact notions, as defined in~\cite{demographic,equalizedOdds}.

For the following definitions, let $U$ be an individual in the dataset, where each individual has data features $(X,A)$. In this context, $A$ denotes the PA and in what follows $A=1$ and $A=0$ express membership of a protected group or not.  $X$ constitutes the rest of attributes that are available to the algorithm. $Y$ denotes the actual label of $U$ while $\hat{Y}$ would be the predicted label by the model: (1) {Demographic Parity (DP)} A classifier satisfies DP when outcomes are equal across groups$F_{DP} = \frac{Pr \left( \hat{Y} = 1 \mid A = 0 \right)}{Pr \left( \hat{Y} = 1 \mid A = 1 \right)}$
(2) \emph{Equalised Odds (EOd)} A classifier satisfies EO if equality of outcomes happens across both groups and true labels: $F_{EOd} = \frac{Pr \left( \hat{Y} = 1 \mid A = 0, Y = \gamma \right)}{Pr \left( \hat{Y} = 1 \mid A = 1, Y = \gamma \right)}$ where $\gamma \in \left\{  0, 1 \right\}$.
(3){Equality of Opportunity (EOp)} is similar to EO, but only requires equal outcomes across subgroups for \emph{true positives}:$F_{EOp} = \frac{Pr \left( \hat{Y} = 1 \mid A = 0, Y = 1 \right)}{Pr \left( \hat{Y} = 1 \mid A = 1, Y = 1 \right)}$

In this paper, we will focus on the computations based on the above three fairness metrics. For this computation, the auditor requires to have access to the three pieces of information for each elements in the dataset: (1) the sensitive group membership (binary value for $A$ demonstrating if a sample does or does not belong to a group with PAs) (2) the actual labelling of the sample (binary value for $Y$) (3) the predicted label of the sample (binary value for $\hat{Y}$). The ML system transfers this information for each sample from their test set. Then, the auditor uses this information to compute the above fairness metrics. 

Note that while we consider the above metrics for our protocol and proof-of-concept implementation in next sections, our core architecture is independent of metrics, and the metric set can be replaced by other metrics too (Fig. \ref{fig:system}).


\subsection{Auditing ML Models for Fairness}
The existing research in fair ML normally assumes the computation of the fairness metric to be done locally by the ML system, with full access to the data, including the private attributes~\cite{equalizedOdds,corbett2017algorithmic,chouldechova2017fair}. However, there is a lack of verifiability and independence in these approaches which will not necessarily lead to trustworthiness.
To increase trust in the ML products, the providers might make the trained model self--explaining~(aka transparent or explainable). There is also the {transparent--by--design} approach \cite{guidotti2018survey,angelino2017learning,wang2017bayesian}. While this approach has its benefits, it is both model--specific and scenario--specific~\cite{panigutti2021fairlens}; thus it cannot be generalised. 
There is also no trusted authority to verify such claims and explanations. 
Moreover, in reality, the trained model, datasets and feature extraction mechanisms are company assets. Once exposed, it can make them vulnerable to the competitors. 
Another approach to provide transparency to the fairness implementation comes through the {black--box auditing}, also known as adhoc ~\cite{guidotti2018survey,lundberg2017unified,panigutti2020doctor}. In this way, the model is trained and audited for different purposes~\cite{adler2018auditing}. This solution is similar to tax auditing and financial ledgers where accountants verify and ensure these calculations are legitimate.
However, unlike the well--established body of certifications and qualifications for accountants in tax auditing and financial ledgers; there does not exist any established processes and resources for fairness computation in AI and ML. 

The concept of a service that calculates fairness has been proposed before, e.g., in \cite{veale2017fairer}. 
The authors introduced an architecture to delegate the computation of fairness to a trusted third party that acts as a guarantor of its algorithmic fairness. In this model, the fairness service is trusted both by the ML system and the other stakeholders (e.g. users and activists). In particular, the ML system must trust the service to maintain the privacy of data and secrecy of its model, whilst revealing to the trusted third party the algorithm outcome, sensitive input data and even inner parameters of the model. 
This is a big assumption to trust that the third party would not misuse the information and hence the leakage of data and model information is not a threat. 

To address these limitations, Kilbertdus et al.~\cite{kilbertus2018blind} proposed a system known as `blind justice', which utilises multi--party computation protocols to enforce fairness into the ML model. Their proposal considers three groups of participants: User~(data owner), Model~(ML model owner) and the Regulator~(that enforces a fairness metric). These three groups collaborate with each in order to train a fair ML model using a federated learning approach~\cite{yang2019federated}. The outcome is a fair model that is trained with the participation of these three groups in a privacy-preserving way. 
They only provide a limited degree of verifiability in which the trained model is cryptographically certified after training and each of the participants can make sure if the algorithm has not been modified. It should be noted that since they operate in the training stage of the ML pipeline, their approach is highly dependent on the implementation details of the ML model itself. 
Jagielski et al.~\cite{jagielski2019differentially} proposed a differential privacy approach in order to train a fair model. Similarly, Hu et al.~\cite{hu2019distributed} used a distributed approach to fair learning with only demographic information. Segal et al.~\cite{segal2021fairness} used similar cryptographic primitives but took a more holistic approach towards the computation and verification of fairness. They proposed a data-centric approach in which the verifier challenges a trained model via an encrypted and digitally certified dataset using merkle tree and other cryptographic primitives. Furthermore, the regulator will certify the model is fair based on the data received from the clients and a set of dataset provided to the model. 
Their approach does not provide universal verifiability as the regulator is the only party involved in the computation of fairness. 
More recently, Park et al.~\cite{park2022fairness} proposed a Trusted Execution Environment~(TEE) for the secure computation of fairness. Their proposal requires special hardware components which are cryptographically secure and provide enough guarantees and verification for the correct execution of the code.

The previous research 
generally has integrated fairness into their ML algorithms; therefore, such algorithms should be redesigned to use another fairness metric set. 
As it can be seen in Table \ref{tbl:comparison}, FaaS is the only work which is independent of the ML model and fairness metric with universal verifiability, and hence, can be used as a service. 



\begin{figure}[t]
	\centering
	\includegraphics[scale=0.64]{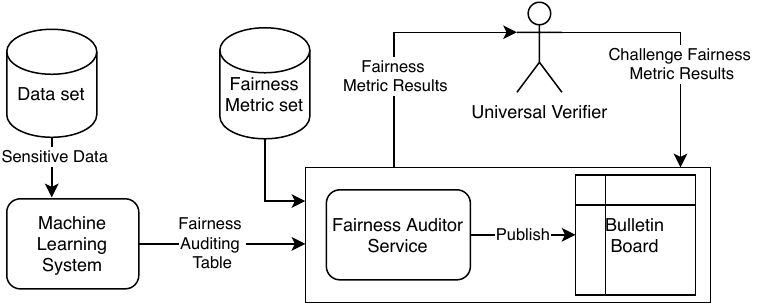}
	\caption{FaaS Architecture}
	\label{fig:system}
\end{figure}

\section{FaaS Architecture}
\label{s:FaaS}
In this Section, we present the architecture of our system (Fig. \ref{fig:system}) and describe its features. The FaaS architecture includes stakeholders in three roles: A) {\bf ML System:} a system that owns the data and the ML algorithm, B) {\bf Fairness Auditor Service:} a service that computes the fair performance of the ML system, and C) {\bf Universal Verifier:} anyone who has the technical expertise and motivation to verify the auditing process.

\subsection{Threat Model}
\label{ss:threat}
The design and implementation of the security of parties implementing the respective protocol roles (ML system, Fairness Auditor Service, and Universal Verifier) (Fig. \ref{fig:system}) are independent of each other. The inter--communications that happen between the roles assumes no trust between parties; thus, all their claims must be accompanied with validation proofs~(for which we will use ZKP). We assume the Auditor System is vulnerable to different attacks and not trustworthy. Thus, the data stored on the Fairness Auditor System must be encrypted, tamper-proof and verifiable at all stages. Moreover, we assume the communication channel between the ML system and fairness auditor is not protected. Therefore, the sensitive data must be encrypted before the transmission starts.  However, there will be an agreement on the cryptographic primitives at the pre--setting stage in the protocol sequence.

In FaaS, we assume that the ML system is honest in sending the cryptograms of the original labels of the dataset samples. One might argue against such assumption and discuss that the ML system might intend to deceive the Auditor Service, and by extension the verifiers, by modifying the actual labels of the dataset. For instance, the ML system would provide the cryptograms of the actual labels and the predicted ones as similar to each other as possible so that the auditor concludes the algorithms are fair. This is an interesting area for further research.  For instance, it may be addressed by providing the cryptograms of the actual labels to the Auditor Service independently e.g. the verifier may own a dataset it provides to a ML system. The verifier then separately decides the desired values for the actual labels and feeds these to the Auditor service.  In this way, it is far less clear to the ML system how to manipulate the data it sends to the auditor, since some of the labels come from elsewhere. 

The internal security of the roles is beyond FaaS. The ML system itself needs to consider extra measures to protect its data and algorithms. We assume the ML system does present the data and predictions honestly. This is a reasonable assumption since the incentives to perform \emph{ethically} is in contrast to being dishonest when participating in fairness auditing process. This is discussed more in the Discussion Section.

\subsection{Protocol Overview}
The main security protocol sequence is between the ML system and Fairness Auditing Service or \emph{auditor} in short form. 
Note that although we suggest three roles in our architecture, the communications are mainly between the above two roles, and any universal verifier can turn to the auditor service (which represents the fairness board), if they want to challenge the computations. 

The ML system is responsible for the implementation and execution of the ML algorithm. It has data as input and performs some prediction~(depending on the use case and purpose) that forms the output~(Fig.~\ref{fig:system}). The Fairness Auditor Service receives information from the ML system, evaluates its fairness performance by computing a fairness metric. Then, it returns the result for the metric back to the ML system. It also publishes the calculations in a \emph{fairness board} for public verification. The public fairness board is a publicly accessible, read-only fairness board (e.g. a website). The auditor only has the right to append data~(and the sufficient proofs) to the fairness board. Also, the auditor verifies the authenticity, correctness and integrity of data before publishing it. 


\subsection{Protocol Sequence}
This protocol has three stages: setup, cryptogram generation and fairness metric computation. 

\begin{table}[t]
\centering
	\caption{\label{tbl:permutations} Possible permutations of 3-bit representation of an entry in the original data.}
\resizebox{0.6\textwidth}{!}{%
	\begin{tabular}{@{}c|c|c|c|c@{}}
		\toprule
		\begin{tabular}[c]{@{}c@{}}Membership\\ of Sensitive Group\end{tabular} & \begin{tabular}[c]{@{}c@{}}Actual\\  Label\end{tabular} & \begin{tabular}[c]{@{}c@{}}Predicted\\ Label\end{tabular} & \begin{tabular}[c]{@{}c@{}}Encoded \\ Permutation\end{tabular} & \begin{tabular}[c]{@{}c@{}}Permutation\\  \#\end{tabular} \\ \midrule
		No                             & 0                              & 0                              & 000                            & \#1                            \\
		No                             & 0                              & 1                              & 001                            & \#2                            \\
		No                             & 1                              & 0                              & 010                            & \#3                            \\
		No                             & 1                              & 1                              & 011                            & \#4                            \\
		Yes                            & 0                              & 0                              & 100                            & \#5                            \\
		Yes                            & 0                              & 1                              & 101                            & \#6                            \\
		Yes                            & 1                              & 0                              & 110                            & \#7                            \\
		Yes                            & 1                              & 1                              & 111                            & \#8                            \\ \bottomrule
	\end{tabular}
}
\end{table}
\normalfont

\subsubsection{Phase I: Setup}
\label{sss:setup}
In this phase, the ML System and Auditor agree on the initial settings. We assume the protocol functions in multiplicative cyclic group setting~(i.e. Digital Signature Algorithm~(DSA)--like group~\cite{katz2014introduction}), but it can also function in additive cyclic groups~(i.e.~Elliptic Curve Digital Signature Algorithm (ECDSA)--like groups~\cite{katz2014introduction}).
The auditor and ML system publicly agree on $(p,q,g)$ before the start of the protocol. Let $p$ and $q$ be two large primes where $q|(p-1)$. In a multiplicative cyclic group~($\mathbb{Z}_{p}^{*}$), $G_q$ is a subgroup of prime order $q$ and $g$ is its generator. For simplicity, we assume the Decision Diffie--Hellman~(DDH) problem is out of scope~\cite{stinson2018cryptography}.

Next, the ML system generates a public/private pair key by using DSA or ECDSA and publishes the public keys in the fairness board. The protection the private key pair depends on the security architecture of the ML system and we assume the private key is securely stored in an industrial standard practice~(e.g. using the secure memory module on board).

{\bf Cryptogram Table:\ } After initial agreements, the ML system produces a cryptogram table with $n$ rows corresponding to the number of samples in their test dataset. We will refer to this table as \emph{cryptogram table} in the rest of this paper. In case the ML system does not want to reveal the number of the samples in the test set, the auditor and the ML system can publicly agree on $n$. In this case, $n$ must be big enough so that the universal verifiers are satisfied with the outcome. 

Each row in the cryptogram table summarises three parameters: (1) protected group membership status, (2) its actual label and (3) predicted label by the ML model. Each row contains the encrypted format of the three parameters along with proofs of its correctness. 
A cryptogram table in the setup phase is shown in Table~\ref{tbl:cryptogram}. 
In the simplest case, each parameter is binary. Therefore, the combined parameters will generate eight permutations in total. In the setup phase, the table is generated to contain all eight possible permutations and their proofs for each data sample. The total structure of the permutations are shown in Table~\ref{tbl:permutations}.
Each row will satisfy four properties: (a) one can easily verify if a single cryptogram is the encrypted version of one of the eight possible permutations, (b) while verifiable, if only one single cryptogram selected, one cannot exert which permutations the current cryptogram represents, (c) for each two cryptograms selected from a single row, anyone will be able to distinguish each from one another, and (d) given a set of cryptograms arbitrarily select from each row as a set, one can easily check how many cases for each ``permutation'' are in the set. 

\begin{table*}[t]
\centering
\footnotesize
	\caption{\label{tbl:cryptogram}Cryptogram Table for $n$ data samples}
 \resizebox{0.9\textwidth}{!}{%
	\begin{tabular}{@{}|c|c|c|c|c|c|c|@{}}
		\toprule
		\begin{tabular}[c]{@{}c@{}}Sample\\  No\end{tabular} & \begin{tabular}[c]{@{}c@{}}Random\\ Public Key\end{tabular} & \begin{tabular}[c]{@{}c@{}}Reconstructed\\ Public Key\end{tabular} & \begin{tabular}[c]{@{}c@{}}Cryptogram \\ of Permutation \#1\end{tabular} & \begin{tabular}[c]{@{}c@{}}Cryptogram \\ of Permutation \#2\end{tabular} & ... & \begin{tabular}[c]{@{}c@{}}Cryptogram \\ of Permutation \#8\end{tabular} \\ \midrule
		1                              & $g^{x_1}$                      & $g^{y_1}$                      & \begin{tabular}[c]{@{}c@{}}$g^{x_1.y_1}.g$, \\ 1-of-8 ZKP\end{tabular} & \begin{tabular}[c]{@{}c@{}}$g^{x_1.y_1}.g^{2^m}$,\\ 1-of-8 ZKP\end{tabular} & ... & \begin{tabular}[c]{@{}c@{}}$g^{x_1.y_1}.g^{2^{7.m}}$, \\ 1-of-8 ZKP\end{tabular} \\ \midrule
		2                              & $g^{x_2}$                      & $g^{y_2}$                      & \begin{tabular}[c]{@{}c@{}}$g^{x_2.y_2}.g$, \\ 1-of-8 ZKP\end{tabular} & \begin{tabular}[c]{@{}c@{}}$g^{x_2.y_2}.g^{2^m}$,\\ 1-of-8 ZKP\end{tabular} & ... & \begin{tabular}[c]{@{}c@{}}$g^{x_2.y_2}.g^{2^{7.m}}$, \\ 1-of-8 ZKP\end{tabular} \\ \midrule
		...                            & ...                            & ...                            & ...                            & ...                            & ... & ...                            \\ \midrule
		n                              & $g^{x_n}$                      & $g^{y_n}$                      & \begin{tabular}[c]{@{}c@{}}$g^{x_n.y_n}.g$, \\ 1-of-8 ZKP\end{tabular} & \begin{tabular}[c]{@{}c@{}}$g^{x_n.y_n}.g^{2^m}$, \\ 1-of-8 ZKP\end{tabular} & ... & \begin{tabular}[c]{@{}c@{}}$g^{x_n.y_n}.g^{2^{7.m}}$,\\ 1-of-8 ZKP\end{tabular} \\ \bottomrule
	\end{tabular}
 }
\end{table*}
\normalfont
The generation of the cryptogram table functions are based on the following sequence:

Step (1): For each of the $n$ samples, the system generates a random public key $g^{x_i}$ where $x_i$ is the private key and $x_i \in \left [  1,q-1 \right ]$. 

Step (2): Once computation of public keys is finished for all samples, the system will compute another number $g^{y_i}$ where computed using Equation below. We refer to as \emph{reconstructed public key} as it is computed using a combination of public keys of all the rows, except for the current one.$g^{y_i} =  \frac{\prod_{j = 1}^{i-1} g^{x_j}}{\prod_{j = i + 1}^{n} g^{x_j}}$.

Step (3): At this step, the ML system computes the cryptograms and zero knowledge proofs for all the possible parameter permutations. This step occurs before the ML system is trained and deployed to predict data samples. Therefore, it considers all the permutation for minimising the overhead in the next protocol sequence stages (as we discuss later).

{\bf Cryptograms: } Each permutation is encoded into a $C_{i}=g^{x_{i}.y_{i}}.g^{p_{i}}$ which are computed based on the multi-option voting schemes introduced in \cite{baudron2001practical} and applied in \cite{hao2014every,hao2010anonymous}. In their method, $p_i$ is computed based on the $n$~(number of samples which already have been publicly agreed) and $m$ as the smallest integer such that $2^m > n$. For each of the eight permutations, the $p_i$ is computed using the following equation:

\begin{equation}
	\label{eq:encoding}
	p_{i} = \left\{\begin{matrix}
		2^0     & for~permutation~\# 1 \\
		2^m     & for~permutation~\# 2 \\
		\cdots  & \cdots               \\
		2^{7.m} & for~permutation~\# 8
	\end{matrix}\right.
\end{equation}

{\bf Zero Knowledge Proofs: } In addition to cryptograms, the ML system also generates 1--out--of--8 ZKP for each of the permutations. This proof ensure the values presented as $C_i$ in the cryptogram table is indeed the production of $g^{x_i.y_i}$ and $g^{p_i}$ where $p_i \in \left \{ 2^0, 2^m, \cdots , 2^{7.m} \right \}$. As shown in Table~\ref{tbl:cryptogram}, each of the computed columns for permutation contains a ZKP to guarantee it is one of the \emph{valid} values for evaluating the fairness metric in next stages. We use the widely used 1--out--of--n interactive ZKP technique \cite{cramer1994proofs}, where $n=8$ in our protocol. Moreover, by application of Fiat--Shamir heuristics~\cite{fiat1986prove}, this ZKP can be converted into non--interactive which makes the verification of proofs simpler~\cite{hao2010anonymous}.

\subsubsection{Phase II: Parameter Assignment}
\label{sss:table}
This stage starts when the ML system's training and testing. The output of this stage is a table with $n$ rows, each containing a cryptogram of the encoded permutation parameters with the required ZKPs, public key~($g^{x_i}$) and reconstructed key~($g^{y_i}$). The outcome of this stage is the final variant of the cryptogram table which we will call~\emph{fairness auditing table}.

{\bf Fairness Auditing Table: } This is derived from the previously computed \emph{cryptogram table}. This table combines the outcome of the ML model~(as shown in encoding format) with the cipher-text created in Phase I and form a ciphered version of the test dataset with $n$ samples. This table is generated based on the following steps:

Step (1): First, the ML system and fairness service properly authenticate each other to ensure they are communicating to the intended party. The ML system determines the permutation combination based on the three items parameters explained before. For that, ML system generates binary encoding for each of the data samples in the test dataset (i.e. the sensitive group membership, actual label and the predicted labels respectively as explained in Table~\ref{tbl:cryptogram}).

Step (2): The ML system generates ZKP for the knowledge of the encoding as commitment to its choice~($p_i$ as in Equation~\ref{eq:encoding}). The ZKP for the proof of knowledge can be converted to non--interactive using Fiat--Shamir heuristic~\cite{fiat1986prove}.

Step (3): The corresponding column number that equals the decimal value of the binary encoding is selected from the cryptogram table to complete the fairness auditing table(~as shown in Table~\ref{tbl:permutations}).

Finally, the generated fairness auditing table is digitally signed by the ML system and then is sent over the Fairness auditing service.

\subsubsection{Phase III: Fairness Evaluation}
\label{sss:fairnesscomputation}
First, the fairness auditing service receives the fairness auditing table, verifies the digital signature and the ZKPs, and publishes the contents in the fairness board. 



Then, it starts the process of computing the fairness metric. For this, the auditor service multiplies all the cryptograms~($C_{i}$) received in the cryptogram table together. Therefore, we have $\prod_{i}{C_i} = \prod_{i}{g^{x_i.y_i}.g^{p_i}}$. At this stage, the key point is the consideration of the effect $y_i$ and $x_i$ have on each other; know as ``Cancellation Formula''~(Lemma~\ref{lema:cancellation} and~\cite{hao2010anonymous,hao2014every,fengVehicle}).

\begin{lemma}
	\label{lema:cancellation}
	\emph{Cancellation Formula: } for $x_i$ and $y_i$, $\sum_{i}{x_i.y_i} = 0$
\end{lemma}

\begin{proof}
	From reconstructed keys equation, one can deduce $y_i$ is as $\sum_{i} = \sum_{j < i}{x_j} - \sum_{j > i}{x_j}$, hence:
	\begin{equation}
		\begin{split}
			\sum_{i}{x_i.y_i} & = \sum_{i=1}^{i=n}{x_i. ( \sum_{j = 1}^{j= i-1}{x_j} - \sum_{j = i+1}^{j = n}{x_j} )}
			\\
			& = \sum_{i = 1}^{i=n}{\sum_{j = 1}^{j= i-1}{x_i.x_j}} - \sum_{i}^{i=n}{\sum_{j = i+1}^{j = n}{x_i.x_j}}
			\\
			& = \sum_{j = 1}^{j=n}{\sum_{i=j+1}^{i= n}{x_i.x_j}} - \sum_{i}^{i=n}{\sum_{j = i+1}^{j = n}{x_i.x_j}}
			\\
			& = \sum_{i = 1}^{i=n}{(\sum_{j=1}^{j= i-1}{x_j}} - \sum_{j=i+1}^{j=n}{\sum_{j = i+1}^{j = n}{x_j})}x_{i}
			\\
			& = 0
		\end{split}
	\end{equation}
	At this point, we expand each of these equation components to compare them together.
\end{proof}



Considering the Cancellation Formula, we can conclude multiplication of all cryptograms into $\prod_{i}{C_i} = \prod_{i}{g^{x_i.y_i}.g^{p_i}} = \prod_{i}{g^{p_i}} = g^{\sum_{i}{p_i}}$. The result is total sum of permutations ($p\#1$ to $p\#8$) as $\sum_{i}{p_i} = a.2^0 + b.2^m + c.2^{2m} + d.2^{3m} + e.2^{4m} + f.2^{5m} + g.2^{6m} + h.2^{7m} $ where $a,b,c,d,e,f,g,h$ are the number of each permutation respectively (Permutation~\#1, Permutation~\#2, $\dots$, Permutation~\#8). 
The search space for such combination depends on the number of samples sent from the ML system to the auditor~(the size of the test set is $n$ for 8 permutations is $\bigl(\begin{smallmatrix}n + 8 -1\\8 - 1\end{smallmatrix}\bigr)$~\cite{hao2010anonymous}). 
As described in Phase I, the size of $n$~(the total number of samples) can be agreed with consideration of the computational capacity of the auditor service. In the simplest setting where $n$ is small, the auditor will determine the overall number of permutations~(as in $\sum{p_i}$, where $i \in \{ 1,2,\cdots,8 \}$) by performing an exhaustive search in all possible combinations until it finds the correct one. 

This process is computationally heavy especially when the number of data samples in the fairness auditing table is large. In this case, the fairness auditor can delegate the declaration of the permutation number to the ML system. The auditor still receives the fairness auditing table and the relevant ZKPs. It can store the fairness auditing table to the fairness board, compute the fairness, and verify the correctness of the declared permutation numbers. The universal verifier can follow the same steps to verify the fairness metric computations through the fairness auditing table that is publicly accessible via fairness board.



At the end of this stage, the auditor uses the acquired numbers to compute the fairness metric and release the information publicly. The number of each permutation denotes the overall performance of the ML algorithm for each of the groups with protected attribute. Table~\ref{tbl:permutation} demonstrates the permutations and how it relates to the fairness metric of the ML system. The cryptogram table and the results will be published on the fairness board  (Fig. \ref{fig:system}).

\section{Implementation and Performance Analysis}
\label{s:evaluation}

\subsection{Proof-of-Concept Implementation}
\textbf{Tools and Platform:} The back--end is implemented in Python v3.7.1 and the front--end is implemented with Node.js v10.15.3. In our evaluations, the computations required for generation of the cryptogram table~(in the ML system) is developed with Python. The elliptic curve operations make use of the Python package \emph{tinyec} and the conversion of Python classes to a JSON compatible format uses the Python package~\emph{JSONpickle}. All the experiments are conducted on a MacBook pro laptop with the following configurations: CPU 2.7 GHz Quad-Core Intel Core i7 with 16 GB Memory running MacOS Catalina v.10.15.5 for the Operating System.

\begin{table}[t]
\footnotesize
\caption{The required permutations to compute the fairness metrics of an ML system}
		\label{tbl:permutation}
\centering
\resizebox{0.5\textwidth}{!}{%
	\begin{tabular}{@{}lcc@{}}
		\toprule
		\multicolumn{1}{c}{\begin{tabular}[c]{@{}c@{}}Fairness\\ Component\end{tabular}} & \begin{tabular}[c]{@{}c@{}}Corresponding\\ Permutation \#\end{tabular} & Computation     \\ \midrule
		$ Pr(\hat{Y} \mid A=0)$                            & \#2 , \#4                      & $(\#2 + \#4)/n$ \\
		$ Pr(\hat{Y} \mid A=1)$                            & \#6 , \#8                      & $(\#6 + \#8)/n$ \\
		$ Pr(\hat{Y} \mid A=0, y=0)$                       & \#2                            & $\#2/n$         \\
		$ Pr(\hat{Y} \mid A=1, y=0)$                       & \#6                            & $\#6/n$         \\
		$ Pr(\hat{Y} \mid A=0, y=1)$                       & \#4                            & $\#4/n$         \\
		$ Pr(\hat{Y} \mid A=1, y=1)$                       & \#8                            & $(\#8)/n$       \\ \bottomrule
	\end{tabular}
}
\end{table}
\normalfont

\textbf{Case-Study Dataset:} We use a publicly available dataset from Medical Expenditure Panel Survey~(MEPS) \cite{MEPS} that contains 15830 data points about the healthcare utilization of individuals. We developed a model (Logistic Regression) that determines whether a specific patient requires health services, such as additional care. This ML system assigns a score to each patient. If the score is above a preset threshold, then the patient requires extra health services. In the MEPS dataset, the protected attribute is ``race''. A fair system provides such services fairly independent of the patient's race. Here, the privileged race group in this dataset is ``white ethnicity''. We have used 50\% of the dataset as training, 30\% as validation and the remaining 20\% as test dataset. 
We set the number of cryptogram table samples to equal the size of test set~($N=3166$). In this example we include three attributes in the cryptogram to represent the binary values of $A$, $Y$ and $\hat{Y}$ (section \ref{ss:fairness}), thus leading to 8 permutations for each data sample.  

In our experiment, where $N=3166$, the total size of the search space is $\bigl(\begin{smallmatrix}3166 + 8 -1\\8 - 1\end{smallmatrix}\bigr) \approx 2^{69}$. The exhaustive search approach is computationally expensive for our experimental hardware configurations, so we decided to use the approach suggested in Section~\ref{sss:fairnesscomputation}. Here, the permutation numbers are declared by the ML system and the auditor service verified the claims by comparing the computations done by the auditor (as in $\prod_{i}{C_i} = \prod_{i}{g^{x_i.y_i}.g^{p_i}} = \prod_{i}{g^{p_i}} = g^{\sum_{i}{p_i}}$) with the total sum of the received permutations ($p\#1$ to $p\#8$) as $\sum_{i}{p_i} = a.2^0 + b.2^m + c.2^{2m} + d.2^{3m} + e.2^{4m} + f.2^{5m} + g.2^{6m} + h.2^{7m}$. This is a reasonable approach since we assumed that the ML system will not attempt to deceive the auditor for its outcome~(section~\ref{ss:threat}).

\subsection{Performance}
\label{ss:time}
This section presents the execution time per data point for each of the main computational tasks, in each protocol stage. Recall that phase I was executed before the ML system's training and testing. This stage can be developed~(and stored separately) in parallel to the implementation of the model in order to mitigate the performance challenge of Phase I. In our implementation, the output of this stage~(cryptogram table) is stored in a  separate file in JSON format and can be retrieved at the beginning of the phase II. 

Phase II begins after the ML model is trained, tested, and validated. This stage uses the output of the ML model to generate the fairness auditing table from the cryptogram table as well as ZKP for knowledge of the permutation. The output of this phase is transmitted to the Fairness Auditor Service in JSON format for phase III. At this stage, first the ZKPs are verified and then, the summation of the cryptograms determines the number of permutations for each of the sensitive groups. Once the auditing service has these numbers, it can compute the fairness of the ML system.

In our evaluations~(where $N=3166$), public/private key pair generation completes in 60 milliseconds~(ms) on average with standard deviation of 6ms. The execution time for ZKP of private key was roughly the same~(60ms on average with standard deviation of 6ms). The generation of reconstructed public key took around 450ms with standard deviation of 8ms. The most computationally expensive stage in phase I was the 1--out--of--8 ZKP for each of the permutations. This stage took longer than the other ones because first, the algorithm is more complicated and second, it should be repeated 8 times~(for each of the permutations separately) for every row in cryptogram table. The computation of 1--out--of--8 ZKPs takes 1.7 seconds for each data sample with STD of 0.1 seconds. Overall, phase I took around 14 seconds with STD of 1 second for each data sample in the test set. In our experiments~(where $N=3166$ samples), the total execution of phase I took roughly 12 hours and 54 minutes.

Phase II consists of creation of the auditing table and generation of the ZKP for knowledge of the permutation. The fairness auditing table is derived from the cryptogram table (as it is mapping the encoding to the corresponding permutation number in the cryptogram table). The elapsed time for such derivation is negligible~(total: 1ms). The generation of ZKP for knowledge of the permutation executed less than 60ms on average with standard deviation of 3ms for each data sample. The completion of both stages took less than 3 minutes. The fairness auditing table is sent to the Fairness Auditor Service for Phase III.

The verification of ZKPs in the last phase~(Phase III) is a computationally expensive operation. The ZKP for the ownership of the private key took around 260ms on average with standard deviation of 2ms. The verification of 1--out--of--8 ZKP for each data point roughly took 2.5 seconds on average with 20ms standard deviation. The verification of the ZKP for knowledge of permutation executed in 100ms with standard deviation of 5ms. The summation of the cryptograms after verification took 450ms overall for $N=3166$ items. 
In our experiment, completion of the stages in phase III took around 2 hours and 30 minutes in total.

In summary, the experimental setup for our architecture, where we computed the required cryptograms and ZKPs for $N=3166$ data points in a real--world dataset, overall time was around 15 hours on the laptop specification given earlier. The main part of the time is consumed by the computation required for phase I~(12 hours and 54 minutes). However, as we noted before, Phase I can be executed before the ML model setup and is stored in a separate JSON file and will be loaded at the beginning of stage II~(after the training and validation of the ML model is complete). The other main computational effort, which can only be done after the ML system's outcomes have been obtained, is in Phase III.  For our example, actual computation of fairness takes two and a half hours. In summary, the creation and handling of cryptograms takes considerable computational effort for realistic datasets and for the fairness metrics that require three attributes.  In what follows we analyse how performance scales with respect to the number of data points as well as with the number of attributes represented in the cryptograms.  

\section{Conclusion}
\label{s:conclusion}
This paper proposes Fairness as a Service (FaaS), a trustworthy service architecture and secure protocol for the calculation of algorithmic fairness. 
FaaS is designed as a service that calculates fairness without asking the ML system to share the original dataset or model information. Instead, it requires an encrypted representation of the values of the data features 
delivered by the ML system in the shape of cryptograms. We used non-interactive Zero Knowledge Proofs within the cryptogram to assure that the protocol is executed as it should.
These cryptograms are posted on a public fairness board for everyone to inspect the correctness of the computations for the fairness of the ML system. This is a new approach in privacy--preserving computation of fairness since unlike other similar proposals that use federated learning approach, our FaaS architecture does not rely on a specific machine learning model or a fairness metric definition for its operation. Instead, one have the freedom of deploying their desired model and the fairness metric of choice.

In this paper we proved that the security protocol guarantees the privacy of data and does not leak any model information. Compared to earlier designs, trust in our design is in the correct construction of the cryptogram by the ML system. Arguably, this is more realistic as a solution than providing full access to data to the trusted third party, taking into account the many legal, business and ethical requirements of ML systems. At the same time, this provides a new challenge in increasing the trust one has in the ML system. 
Increasing trust in the construction of the cryptograms remains an interesting research challenge following from the presented protocol. 

We implemented a proof-of-concept of FaaS and conducted performance experiments on commodity hardware. The protocol takes seconds per data point to complete, thus demonstrating in performance challenges if the number of data points is large (tens of thousands). To mitigate the performance challenge, the security protocol is staged such that the construction of the cryptogram can be done off-line. The performance of the calculation of fairness from the cryptogram is a challenge to address in future work. All together, we believe FaaS and the presented underlying security protocol provide a new and promising approach to calculating and verifying fairness of AI algorithms.  

\section*{Acknowledgement}

The authors in this project have been funded by UK EPSRC grant ``FinTrust: Trust Engineering for the Financial Industry'' under grant number EP/R033595/1, and UK EPSRC grant ``AGENCY: Assuring Citizen Agency in a World with Complex Online Harms'' under grant EP/W032481/1 and PETRAS National Centre of Excellence for IoT Systems Cybersecurity, which has been funded by the UK EPSRC under grant number EP/S035362/1.

\bibliographystyle{plain}
\bibliography{references}
\end{document}